\documentclass[12pt,draftclsnofoot,onecolumn]{IEEEtran}
\usepackage{setspace}
\usepackage{amsthm}
\usepackage{amsmath}
\usepackage{amssymb}
\usepackage{epsfig}
\usepackage{amsmath,epsfig}
\usepackage{amssymb, amsthm}
\usepackage{mathrsfs}
\usepackage{algorithm,algorithmic}
\usepackage{enumerate}
\usepackage{amsthm}
\usepackage{amsmath}
\usepackage{amssymb}
\usepackage{epsfig}

\usepackage{paralist}
\usepackage{graphicx}
\usepackage{cite}

\newtheorem{myth}{\bf Theorem}

\newtheorem{myle}[myth]{\bf Lemma}
\newtheorem{mydef}{\bf Definition}

\newcommand{\bx}{\boldsymbol x}

\newcommand{\bg}{\boldsymbol g}

\newcommand{\mK}{\mathcal{K}}
\newcommand{\mR}{\mathcal{R}}
\newcommand{\mX}{\mathcal{X}}

\newcommand{\mT}{\mathcal{T}}

\begin{document}
\sloppy \hyphenation{net-works}
\title{Capacity of Wireless Distributed Storage Systems with Broadcast Repair
\thanks{Part of this work was presented in ICICS 2015~\cite{myicics15}.}\thanks{This work was supported in part by a grant from the University Grants
Committee of the Hong Kong Special Administrative Region, China, under
Project AoE/E-02/08, and in part by ARC Discovery Project DP150103658.}}
\author{\large Ping Hu, Chi Wan Sung, Terence H. Chan}
\long\def\symbolfootnote[#1]#2{\begingroup
\def\thefootnote{\fnsymbol{footnote}}\footnote[#1]{#2}\endgroup}
\maketitle
\pagestyle{empty}
\thispagestyle{empty}

%
%
%

\begin{abstract}
In wireless distributed storage systems, storage \mbox{nodes} are
connected by wireless channels, which are broadcast in nature.
This paper exploits this unique feature to design an efficient
repair mechanism, called broadcast repair, for wireless
distributed storage systems in the presence of multiple-node failures. Due to the broadcast nature of wireless transmission, we advocate a
new measure on repair performance called \emph{repair-transmission
bandwidth}. In contrast to repair bandwidth, which measures the
average number of packets downloaded by a newcomer to replace a
failed node, repair-transmission bandwidth measures the average
number of packets transmitted by helper nodes per failed node. A
fundamental study on the storage capacity of wireless distributed
storage systems with broadcast repair is conducted by modeling the
storage system as a multicast network and analyzing
the minimum cut of the corresponding information flow graph. The fundamental
tradeoff between storage efficiency and repair-transmission bandwidth
is also obtained for functional repair. The performance of broadcast repair is compared
both analytically and numerically with that of cooperative repair,
the basic repair method for wired distributed storage systems with
multiple-node failures. While cooperative repair is based on the
idea of allowing newcomers to exchange packets, broadcast repair
is based on the idea of allowing a helper to broadcast
packets to all newcomers simultaneously. We show that
broadcast repair outperforms cooperative repair, offering a better tradeoff between
storage efficiency and repair-transmission bandwidth.
\end{abstract}


\section{Introduction}
Distributed storage systems (DSS) enable users to access data reliably
anywhere and anytime. It has received more and
more interests for its wide applications, such as cloud
computing, file sharing, and peer-to-peer systems.
Since storage nodes may fail at times, reliability should be ensured so that a user can retrieve his/her stored files even when some nodes are not available.
Common ways to provide
reliability are to use \emph{repetition codes} and \emph{erasure
codes}. Repetition codes have been widely used in DSS systems, such as the Google file system~\cite{googlefile}.
Erasure codes are more space efficient than repetition codes~\cite{compare} and have been employed by OceanStore~\cite{oceanstore} and TotalRecall~\cite{totalrecall}.

Apart from reliability, repairability is another important design issue of a DSS.
Since storage node failures are not rare, it is important that failed nodes
can be repaired in an efficient manner. Traditional erasure codes, while having high
storage efficiency, typically require a large amount of data exchange, called {\em repair bandwidth},
during node repair. The replacement node, commonly called the \emph{newcomer}, needs to download the whole file from some or all of the surviving nodes called \emph{helper nodes}, to reconstruct the data the failed node originally stored.
In this case, the repair bandwith is the whole file size.
Dimakis \emph{et al}. showed that repair bandwidth can be reduced at the expense of storage space~\cite{Dimakis}.
By using \emph{information flow graph},
the repair dynamics of a DSS is modeled as a multicast network. The storage
capacity of a DSS, i.e., the maximum file size that can be stored, is shown to be equal to the min-cut of the information flow graph~\cite{networkcoding}. Given a fixed file size, the fundamental tradeoff between
storage efficiency and repair bandwidth has been derived. Codes that achieve optimal tradeoff are named \emph{regenerating codes}. Specifically, on the two extremal points of the curve, codes attaining the best storage efficiency are called minimum-storage regenerating (MSR) codes whereas codes attaining the minimum repair bandwidth are called
minimum-bandwidth regenerating (MBR) codes.
Furthermore, the optimal storage-repair bandwidth tradeoff can be achieved by linear network
codes with finite alphabet even though the information flow graph is unbounded~\cite{wuyunnan_jsac}. These fundamental results are re-examined and proved in~\cite{itw15}.

The seminal work of~\cite{Dimakis} has stimulated a lot of study
on efficient repair of failed nodes in DSS.
Many codes with different repair characteristics have been constructed.
Basically, there are two repair mechanism, functional repair and exact repair.
In \emph{functional repair}~\cite{Dimakis}, the
newcomer is not required to store exactly the same symbols as
the original failed node as long as the data collector is able to retrieve the file.
In \emph{exact repair}~\cite{exact_isit10,exact_repair,exact_erasure,exact_repair12,uncoded}, the repair process should enable the repaired node to reconstruct the same data stored on the failed node.
Most of these works, however, focus on \emph{single-node repair}, which
means that nodes are assumed to be failed one by one and the
repair process is triggered immediately when a node fails. In~\cite{Cooperative_hu}, it was observed that repair
bandwidth per failed node can be reduced if the repair process is
postponed and triggered only when the number of failed nodes
reaches a predetermined threshold. In this mechanism, newcomers first download some data from the surviving nodes, and then exchange some data among themselves. It is termed {\em
cooperative repair}. Coding methods under specific cases for cooperative repair are proposed in~\cite{cooperative_MFR}\cite{cooperative_function_code}. In~\cite{kenneth_jnl}, results on cooperative
repair are further extended to a more general scenario, and the
fundamental tradeoff in a DSS with cooperative repair is derived. Furthermore, exact cooperative regenerating codes at the minimum bandwidth cooperative regeneration point for all possible parameters are constructed explicitly in~\cite{cooperative_exact_code}.

Nowadays due to the increasing use of wireless devices and the popularity of
wireless sensor networks, the design and analysis of wireless distributed
storage systems (WDSS) has become an emerging new area~\cite{access_sensor,wireless_network_coding,wireless_awgn,wireless_download,Energy_jnl,wireless_backup}.
Recently, in the context of mobile cellular systems, it has been proposed to store or cache popular files in the wireless edge by deploying a number of small-cell base stations with large storage capacity~\cite{femtocaching}. Since the storage capacity of modern smart phones or tablets is ever increasing, it is also possible to exploit this capability and store files directly in these devices. When a user requests a file, he or she can retrieve it by downloading data from neighboring devices through device-to-device (D2D) communications~\cite{d2d_2014,wireless_d2d,d2d_2016}. In such a scenario, it is necessary to repair lost data when a device becomes unavailable such as running out of battery or moving out of the area.

The repair problem in WDSS has been investigated in~\cite{two_layer,repair_enegry,repair_space_time,xiaoming,xiaoming14,xiaoming_broadcast,wireless_d2d,d2d_2016,mingjun2017}.
Among them, the works closest to ours are~\cite{xiaoming,xiaoming14,xiaoming_broadcast}.
In~\cite{xiaoming}\cite{xiaoming14}, the fundamental storage-bandwidth tradeoff for
single-node repair in DSS with erasure channels are established.
In~\cite{xiaoming_broadcast}, it considered the repair problem when parts of stored packets in multiple nodes are lost. The channel for repair is assumed broadcast without interference.
The work focused on one repair round and obtained the minimum transmitted packets for repair. For a special parameter setting, an exact repair code construction is proposed.
Repair under multiple repair rounds is unclear for general parameter settings.

Concerning the repair problem, while designs for DSS can also be
applied to WDSS, it is important to understand the fundamental
difference between DSS and WDSS.
One basic characteristic of the wireless medium, which
distinguishes it from wired transmission, is its {\em broadcast}
nature.
To design an efficient WDSS, the broadcast nature of the wireless medium
could be exploited during the repair process for multiple failed nodes.
To see this, consider the example shown in
Figure~\ref{fig:broadcast_repair}. There are 4 storage nodes. The
file, which is divided into $A_1,A_2,B_1$, and $B_2$, can be
retrieved from any two storage nodes. Suppose two of the nodes
fail as indicated. To replace them, two newcomers join the system,
and the two surviving storage nodes broadcast packets to them.
As shown in the figure, by broadcasting 4 packets in total,
the content originally stored in the failed nodes can be
regenerated by the newcomers. The average number of packet
transmissions needed per newcomer is 2. Suppose we
repair the failed nodes one by one using unicast transmission. The packet would be counted twice, which is 8, in calculating
repair bandwidth. Each newcomer requires 4 packets for repair.
Furthermore, as shown in Fig. 2, if we repair the failed nodes by cooperative repair~\cite[Section I]{kenneth_jnl}, the transmitted packets we need is~6, meaning that the average number of packet
transmissions needed per newcomer is~3.
From this example, we see that the number of packet transmitted over the air can be reduced by broadcast.

\begin{figure}
 \centering
 \includegraphics[width=5.5in,angle=0]{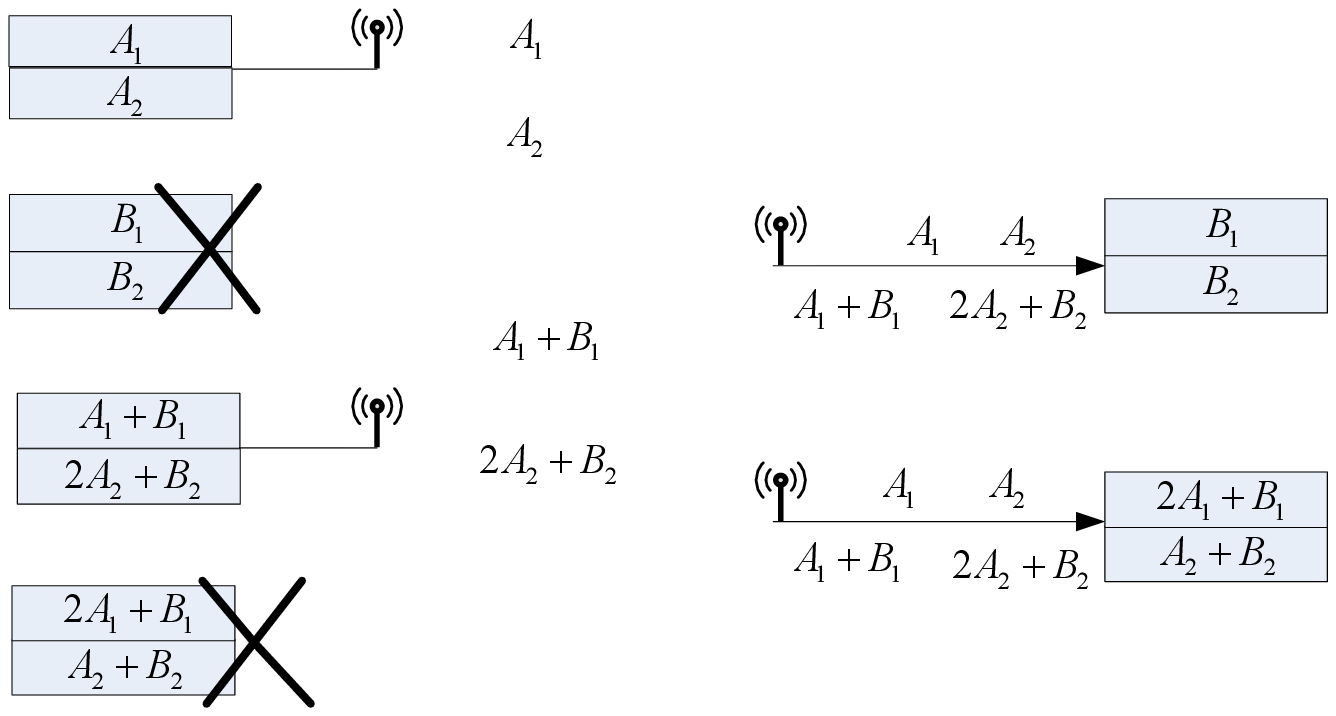}
 \caption{An example for broadcast repair in WDSS.}~\label{fig:broadcast_repair}
\end{figure}

\begin{figure}
 \centering
 \includegraphics[width=5.5in,angle=0]{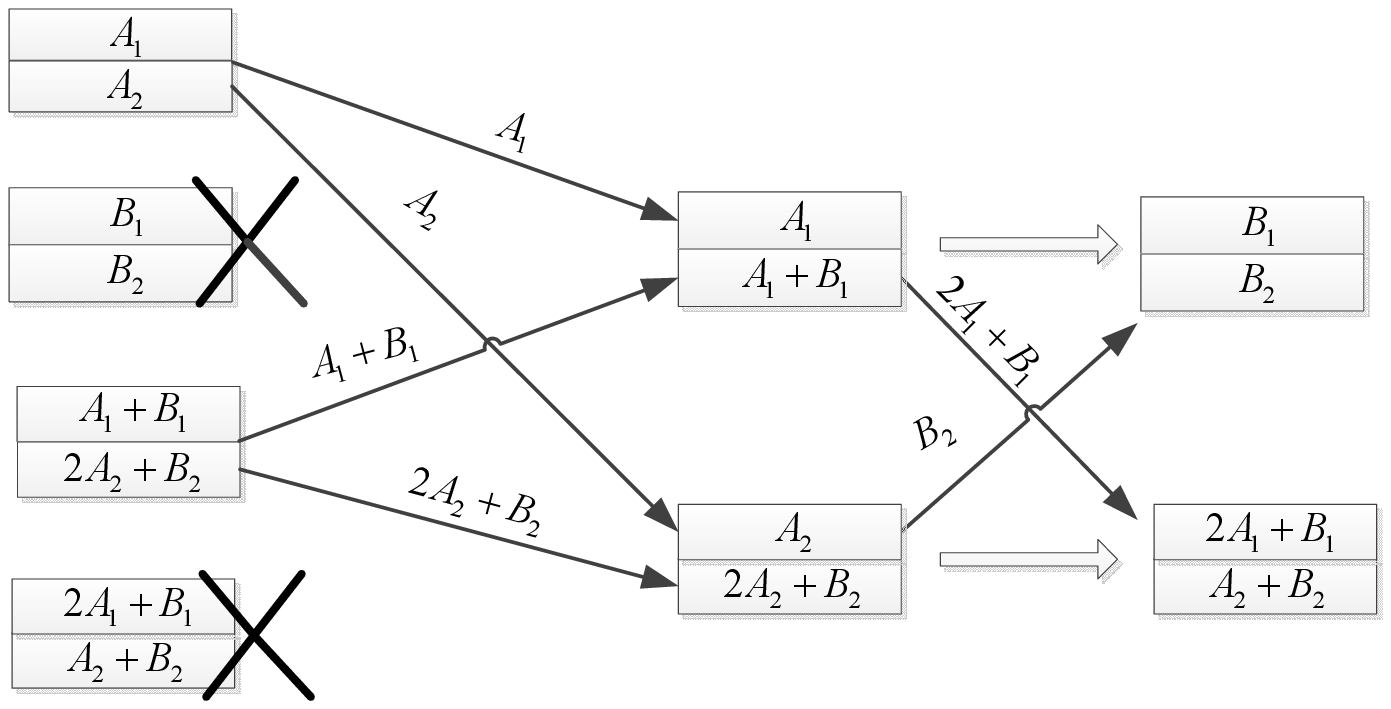}
 \caption{An example for cooperative repair~\cite[Section I]{kenneth_jnl}.}~\label{fig:cooperative_repair}
\end{figure}

To reap the potential gain, we propose the concept of {\em broadcast repair} for WDSS with multiple node failures. We focus on functional repair in this work, so that a graph representation for WDSS is constructed. By analyzing the min-cut of the graph, the storage capacity under finite number of repairs is derived.
The storage capacity obtained is a monotonic decreasing function of the number of repair rounds, but it becomes a constant when the number of repair rounds increases beyond a certain threshold.
We also prove that the storage capacity is achievable by a code with finite alphabet even if there are infinite repair rounds for functional repair.
To quantify the benefit of broadcast repair, we compare our method with cooperative
repair with unicast transmissions. An explicit
form on the storage capacity is derived, and its superiority over cooperative repair for WDSS
is shown analytically. 

Finally, it should be noted that our model for broadcast repair, studied in our conference version~\cite{myicics15}, is equivalent to the model for centralized repair, investigated in~\cite{centralized_repair,centralized_repair_j}. The work in~\cite{centralized_repair,centralized_repair_j} focuses on the minimum-storage point and the minimum-bandwidth point. It is shown that the minimum-storage point under centralized repair can be achieved by the same codes that achieve the minimum-storage point under cooperative repair. Besides, explicit code constructions at this point with exact repair are considered in~\cite{centralized_repair,centralized_repair_j,highrate_MDS}. At the minimum-bandwidth point, exact-repair codes for systems that satisfy a certain technical property are constructed in~\cite{centralized_repair,centralized_repair_j}. In this paper, we characterize the entire storage-bandwidth tradeoff curve, rather than only the two extreme points. Moreover, our approach, based on analyzing the information flow graph, is different from the study in~\cite{centralized_repair}.

Our paper is organized as follows. In Section~\ref{sec:model}, we
describe the system model. In Section~\ref{sec:graph}, we
represent the WDSS by a graph. In Section~\ref{sec:finite rounds}, we
derive the min-cut value of the graph, which serves as an upper bound of the storage capacity.
In Section~\ref{sec:inf_rounds}, we show that the min-cut value is achieveable under infinite repair rounds.
We derive the explicit form of broadcast repair storage capacity
and compare its performance with cooperative repair both analytically and numerically in Section~\ref{sec:compare}.
Section~\ref{sec:conclu} concludes this paper.

\section{System Model and Broadcast Repair}\label{sec:model}
The system model is designed to capture the broadcast
characteristic under the wireless scenario. It includes one source
node, multiple storage nodes, and multiple data collectors. Each
storage node can store $\alpha$ packets at most. Storage nodes are not point-to-point connected, but fully connected by a wireless broadcast medium.

At the initial stage, the source node distributes a file into $n$
storage nodes such that the data collectors can
retrieve the file from any $k$ nodes. We
index these storage nodes by the set $\mathcal{N}\triangleq
\{1,2,\dots,n\}$. After the initialization, the source node
becomes inactive and leaves the system. In other words, the initialization is simply writing data onto the storage nodes.

These $n$ storage nodes are not reliable and can fail at times (becoming inactive). When the number of failed nodes
accumulates up to a threshold $r$, the repair
process is triggered. We call this process one {\em round} of repair. During each
repair round, $r$ new nodes, called \emph{newcomers}, join
the system to replace the failed nodes. To regenerate the lost data, any $d\geq k$ active nodes can be chosen as
\emph{helper nodes} to broadcast packets to the newcomers. The number of packets broadcast by each helper is denoted by $\beta$. To ensure enough helper nodes, we require
\begin{equation} \label{nrd}
n-r\geq d.
\end{equation}
When a packet is broadcast for data repair, we assume that it can be received successfully by all
newcomers without error. Besides, we also assume that the helpers use orthogonal channels for transmissions so that there is no interference between them. We index the
newcomers after the $s$-th repair round by
$\mathcal{R}_s\triangleq\{n+(s-1)r+1,\dots,n+sr\}$. The set of
helpers in this repair round is denoted by $\mathcal{H}_s$.

Any data collector can join the system after the initialization
stage or after any repair round. It can connect to any $k$
active nodes via $k$ orthogonal channels. We assume that the data collector
retrieves all data stored in the nodes it contacts, which can be modeled by
assuming each channel has infinite capacity. In practice, it is possible that a data
collector joins the system when some nodes are failed but the next repair round, say $s+1$, has not been triggered or finished.
Because of~\eqref{nrd}, there are always $k$ or more active nodes in the system for the data collector to retrieve the file.
Such a data collector can be regarded as joining the system right after the $s$-th round. Therefore, without loss of generality, we assume that all data collectors join the system right after the completion of a repair round.

For ease of presentation, we call the initialization stage the 0-th repair round. We denote the data collector which joins after the $s$-th repair round and connects to a set $\mK$ of $k$ active nodes by $\mathsf{DC}_{s,\mK}$.
Since we have to ensure that the file can always be retrieved, we consider all possible arrivals (in terms of $s$) and connections (in terms of $\mK$) of a data collector.

Denote the total number of repair rounds by $T$. The above system is called a WDSS with parameters
$(n,k,d,r,\alpha,\beta,T)$. An \emph{instance} of a
WDSS is determined by the failure patterns, newcomers
$\mathcal{R}_1,\mathcal{R}_2,\dots,\mathcal{R}_T$, and the
collection of helper sets
$\mathcal{H}_1,\mathcal{H}_2,\dots,\mathcal{H}_T$. The repair process described
above is called \emph{broadcast repair}. Denote the storage capacity, which is the maximum file size that can be supported by $\text{WDSS}(n,k,d,r,\alpha,\beta,T)$ under broadcast repair, by
$C_{\text{storage}}^T$. When the WDSS is used indefinitely, so that there is no restriction on the number of repair rounds, we denote the maximum file size that can be supported by $C_{\text{storage}}$. It is clear that $C_{\text{storage}} \leq C_{\text{storage}}^T$ for any $T$.


In the literature of DSS, the total number of packets downloaded
by a newcomer so as to repair a failed node is called {\em repair
bandwidth}~\cite{Dimakis}. It is one of the key performance
metrics in DSS, reflecting the amount of network traffic
required in the repair process. The same concept can also be applied to
multiple node failures with cooperative repair
processes~\cite{Cooperative_hu,cooperative_MFR,cooperative_function_code,kenneth_jnl}. In a wireless environment,
repair bandwidth, however, is not an accurate measure on network traffic, especially when there are multiple node failures.
To better reflect the use of frequency spectrum
in a wireless environment, we introduce a new performance metric
named \emph{repair-transmission bandwidth}:
\begin{mydef}
 The repair-transmission bandwidth, $\tau$, is defined as the number of packets the helper nodes transmitted per newcomer.
\end{mydef}

If all the packet transmissions are in unicast mode, then
repair-transmission bandwidth is equal to repair bandwidth, since
the total number of packets transmitted by the helpers is equal
to the total number of packets received by the newcomers. They are
different, however, when packets transmissions are in broadcast mode. For
the WDSS model described above, we have
$$
\tau=\frac{d\beta}{r}.
$$

When $r\geq k$, the $d$ helpers should transmit the whole file to the $r$ newcomers
because if a data collector connects to $k$ of these $r$ newcomers, it
should be able to retrieve the file. For this case, we can directly obtain
$$
\tau=\frac{C_{\text{storage}}}{r},
$$
no matter how large the storage space $\alpha$ of a node is.
Each newcomer can then reconstruct the original file and re-encode it in the same way as the source had done in the initial stage. This corresponds to exact repair, and obviously the system can sustain for infinite repair rounds.

When $r < k$, it may not be necessary for the helpers to transmit the whole file to the $r$ newcomers. In this paper, we consider only this non-trivial case. We will see that given a requirement on storage capacity $C_{\text{storage}}$,
there is a tradeoff between the per-node storage capacity, $\alpha$, and the repair-transmission bandwidth, $\tau$. For easy reference, we summarize our notation in Table~I.

\begin{table*}[htbp]\label{table:notation_model}
   \renewcommand{\arraystretch}{1.3}
   \caption{Notation summary}
   \label{tab}
   \centering
   \begin{tabular}{|c|c|}
     \hline
     Symbol & Definition\\
     \hline
     $n$ & number of storage nodes \\
     \hline
     $k$ & minimum number of storage nodes required for file reconstruction\\
     \hline
     $s$ & repair round index\\
     \hline
     $T$ & total number of repair rounds\\
     \hline
     $d$ & number of helpers in each repair round\\
     \hline
     $r$ & number of newcomers in each repair round\\
      \hline
     $\alpha$ & per-node storage capacity\\
      \hline
     $\beta$ & number of transmitted packets of each helper \\
      \hline
     $\tau$ & number of packets the helpers transmitted per newcomer \\
      \hline
     $ \mathcal{R}_s$ & the set of newcomers in repair round $s$\\
      \hline
     $ \mathcal{H}_s$ & the set of helpers in repair round $s$\\
      \hline
   \end{tabular}
 \end{table*}

\section{Graph Representation}\label{sec:graph}

\begin{figure*}
 \centering
 \includegraphics[width=6.5in,angle=0]{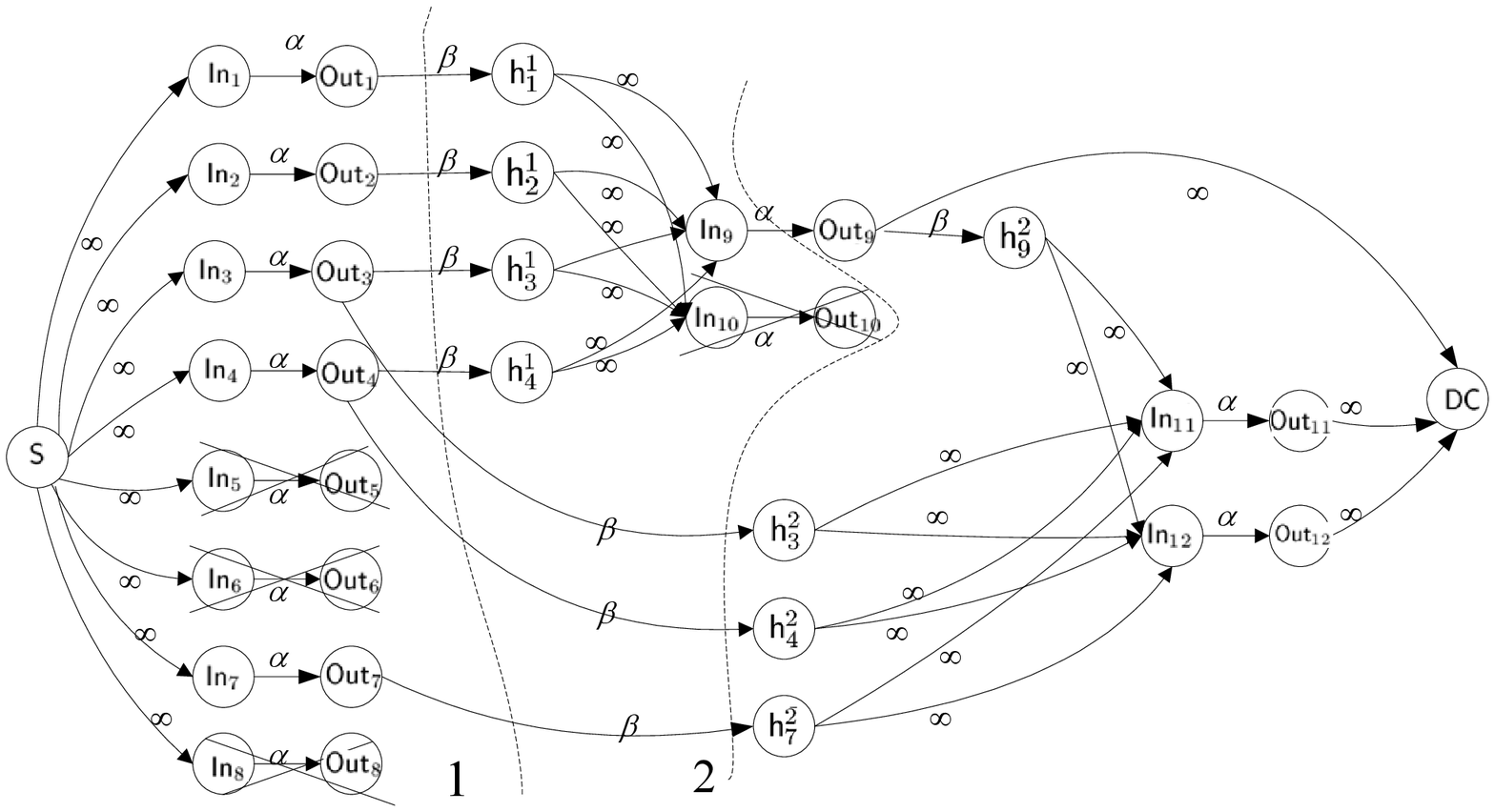}
 \caption{An example for cut-value in broadcast repair.}~\label{fig:define_example}
\end{figure*}

The repair dynamics of a WDSS can be represented by a directed acyclic graph (DAG) $G=(V,E)$, where $V$ is the vertex set and $E$ is the edge set.
Each edge $e(i,j) \in E$, which connects node~$i$ to node~$j$, is associated with a parameter $u_{i,j}$,
which denotes the capacity of the edge.
The graph includes one source vertex $\mathsf{S}$, multiple
storage nodes, and multiple
data collectors $\mathsf{DC}_{s,\mathcal{K}}$'s. Each
storage node $j$ is represented by two vertices, \emph{in-vertex} $\mathsf{In}_j$,
\emph{out-vertex} $\mathsf{Out}_j$, and a directed edge
$\mathsf{In}_j\to\mathsf{Out}_j$ with parameter $\alpha$.
In this paper, the terms ``node'' and ``vertex'' have different
meanings. A node refers to a storage device in the WDSS while a vertex is an abstract entity in
the graph.

In the initialization stage where data is first stored at the storage nodes, the source vertex $\mathsf{S}$
transmits packets to the storage nodes and then becomes inactive. This is
modeled by adding the edges $\mathsf{S}\to\mathsf{In}_j$, for all
$j \in \mathcal{N}$, with capacity $\infty$. Note that this does not mean that the actual transmission rate of the communication link between the source and each storage node is infinite; it only means that all the information in the source link are available in the in-vertices of each storage node, and each storage node can store only $\alpha$ symbols as indicated by the edge capacity of $\alpha$ between the in-vertex and the out-vertex of a storage node.

In the first repair round (i.e., $s=1$), node $i\in\mathcal{H}_1$ broadcasts $\beta$ packets to newcomer $j\in\mathcal{R}_1$, which is again modeled by two vertices $\mathsf{In}_j$, $\mathsf{Out}_j$, and a directed edge $\mathsf{In}_j\to\mathsf{Out}_j$ with parameter $\alpha$.
Note that $\mathcal{R}_1$ and $\mathcal{N}$ are disjoint, meaning that a newcomer has a new index, which is different from the index of the failed node being replaced by that newcomer.
For each helper node $i\in\mathcal{H}_1$, we add an \emph{auxiliary vertex}, say $\mathsf{h}_i^1$. This auxiliary vertex is used to model the broadcast feature of the wireless channel.
Auxiliary vertex $\mathsf{h}_i^1$ outgoes from $\mathsf{Out}_i$  by an edge with capacity $\beta$.
Edges with capacity $\infty$ are added from vertex $\mathsf{h}_i^1$ to in-vertex $\mathsf{In}_j$ of every newcomer $j \in \mathcal{R}_1$.
Subsequent repair rounds are modeled in the same way.

Consider the example shown in Fig.~\ref{fig:define_example}. The corresponding WDSS has parameters $n=8, k=3, d=4$, $r=2$ and $T=2$. In this example, nodes $5$ and $6$ failed in the first repair round, and we have $\mathcal{R}_1=\{9,10\}$ and $\mathcal{H}_1=\{1,2,3,4\}$.
Nodes $8$ and $10$ failed in the second repair round, and we have $\mathcal{R}_2=\{11,12\}$ and $\mathcal{H}_2=\{9,3,4,7\}$.


To model the file retrieval process, after each repair round~$s$ and
for each possible choice of $\mK$, we add a data collector $\mathsf{DC}_{s,\mK}$.
Furthermore, a directed edge from each out-vertex of a node in $\mK$ to
$\mathsf{DC}_{s,\mK}$ with capacity $\infty$ is
added. In Fig.~\ref{fig:define_example}, we show only one data collector, namely, $\mathsf{DC}_{2,\{9,11,12\}}$,
for simplicity.

An $x$-$y$ {\em cut} $\mX$ is a subset of $V$ such that $x\in \mathcal{X}$, $y \in \overline{\mX}\triangleq V \setminus \mX$ and there is at least one edge from $\mX$ to $\overline{\mX}$. The {\em cut-set} of a cut $\mX$ is $\{(u,v) \in E : u\in \mX, v \in \overline{\mX}\}$.
The {\em cut-value} of $\mX$ is defined as:
\begin{equation}
C(\mX)\triangleq\sum_{i\in\mX,j\in\overline{\mX}} u_{ij}.\label{eq:def_cut}
\end{equation}
Two examples of $\mathsf{S}$-$\mathsf{DC}_{2,\{9,11,12\}}$ cuts are denoted in
Fig.~\ref{fig:define_example} by dashed lines.
For line~$1$, the cut-value is $7\beta$, which means that the information that can pass through
this cut is at most $7 \beta$. Note that this is just an upper bound, as the information actually transmitted
over these seven edges can be correlated. Similarly, for line~$2$, the cut-value is $ \alpha+3\beta.$


A WDSS instance $I$ is specified by the failed nodes and the helpers in each repair round. Given any instance $I$, we can construct a graph as described above, which corresponds to a multicast network problem with the single source $\mathsf{S}$ and multiple destinations $\mathsf{DC}_{s,\mK}$, where $s= 0, 1, \ldots, T$, and $\mK$ can be any legitimate choice of storage nodes after repair round~$s$.

According to~\cite{networkcoding}, the capacity of the single-source multicast network is given by the minimum cut-value between the source node and any of the destinations. Therefore, the storage capacity of a particular WDSS instance $I$ with $T$ repair rounds is given by
\begin{align}
C_{\text{storage}}^T(I) = \min_{\mathsf{DC}} \min_{\mX:\mathsf{S}-\mathsf{DC}\text{~cut}}C_I(\mX),\label{eq:def_capa_instance}
\end{align}
where the first minimum is taken over all legitimate choices of $\mathsf{DC}$ under the instance $I$.
This capacity can be achieved by random linear network coding~\cite{random_linear_network}. Since the coding is random, it does not rely on the structure of the network. If the code rate is set to be the minimum value of $C_{\text{storage}}^T(I)$ over all
possible WDSS instances, then that code rate is achievable in all instances.
The storage capacity, $C_{\text{storage}}^T$, of a WDSS with $T$ repair rounds  can be
obtained by minimizing $C_{\text{storage}}^T(I)$ over all its possible instances, i.e.,
\begin{align}
C_{\text{storage}}^T = \min_{I} C_{\text{storage}}^T(I).\label{eq:def_capa}
\end{align}

By definition, it is clear that $C_{\text{storage}}^T$ is monotonic decreasing with $T$. Since zero is a lower bound, the sequence $C_{\text{storage}}^1, C_{\text{storage}}^2, \ldots$ converges to a limit, which we denote by $C_{\text{storage}}^\infty$. Since each term of the sequence is an upper bound of $C_{\text{storage}}$, we must have $C_{\text{storage}} \leq C_{\text{storage}}^\infty$. In the next section, we will show that $C_{\text{storage}}^T = C_{\text{storage}}^\infty$ for all $T \geq k$ and explain how the value of $C_{\text{storage}}^\infty$ can be found by solving a combinatorial optimization problem. In the section after next, we will show that $C_{\text{storage}}$ is indeed equal to $C_{\text{storage}}^\infty$.

\section{Min-Cut Upper Bound of the Storage Capacity}\label{sec:finite rounds}

In this section, we provide an upper bound of the storage capacity by investigating the min-cut value of the family of graphs for a fixed, finite value of $T$. To find the min-cut value, we provide a lower bound in Theorem~\ref{th:bound} and show the tightness of this bound in Theorem~\ref{th:mincut}. For a WDSS with a finite number of repair rounds, this min-cut value can be achieved by linear network coding. For a WDSS which is required to tolerate an infinite number of repair rounds, this min-cut value serves as an upper bound of the storage capacity.



To find a lower bound on the min-cut value for a graph with $T$ rounds, we need to examine all the cuts in the graph.
For any cut, we have the following result:

%


\begin{myth} \label{th:bound}
Consider a graph with $T$ repair rounds, where $T\geq k$. For any $\mathsf{S}-\mathsf{DC}$ cut $K$ in the graph, the cut-value $C(K)$ is bounded below by
\begin{align}
B\triangleq \min_{\bx,\mT_1}\left\{x_0\alpha+\sum_{s\in \mT_1}x_s\alpha+\sum_{s\in \mT_2}(d-\sum_{i=0}^{s-1}x_i)\beta\right\},\label{eq:bound}
\end{align}
where the minimization is taken over $\mT_1\subseteq \mK \triangleq \{1, 2, \ldots , k\}$ and
\begin{align}
&0 \leq x_0 \leq n, \label{eq:x0}\\
  &0 \leq x_s\leq r, \text{ for } s\in \mathcal{K}, \label{eq:xs_range} \\
  &x_0 + \sum_{s\in\mK}x_s= k. \label{eq:sum_xs}
\end{align}
\end{myth}

\begin{figure}
 \centering
 \includegraphics[width=4in,angle=0]{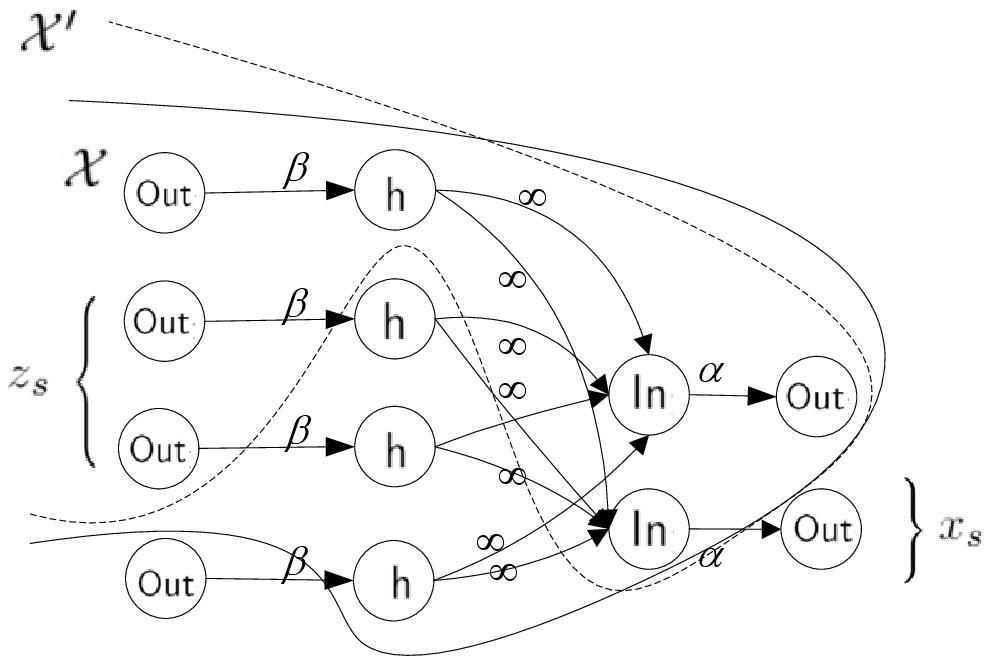}
 \caption{An illustration of case where $s\in \mT_1$.}~\label{fig:analysis_T1}
\end{figure}
\begin{figure}
 \centering
 \includegraphics[width=4in,angle=0]{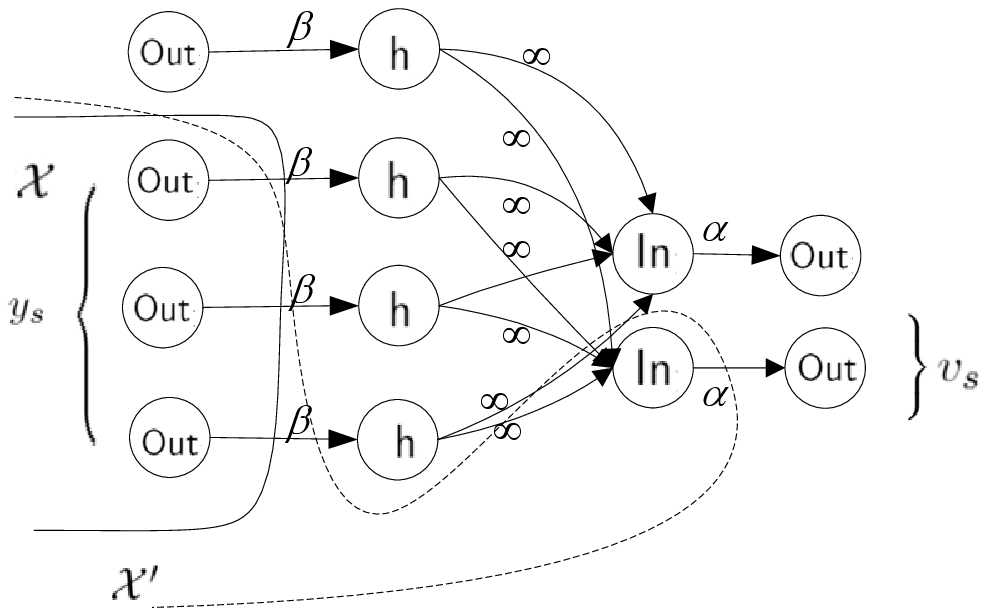}
 \caption{An illustration of case where $s\in \mT_2$.}~\label{fig:analysis_T2}
\end{figure}

\begin{proof}
Consider an arbitrary instance of the WDSS. Regard the initialization stage as round 0 and let $\mathcal{V}_0\triangleq \mathcal{N}$. For rounds $s=1,2,\ldots,T$, let $\mathcal{A}_s \triangleq \{\mathsf{h}_i^s : i\in\mathcal{H}_s\}$ be the set of auxiliary vertices in round~$s$, and $\mathcal{V}_s\triangleq \mathcal{A}_s \cup \{\mathsf{In}_j, \mathsf{Out}_j : j \in \mR_s\}$ be the set of all vertices in round~$s$. Then $\mathcal{V}_0 \cup \mathcal{V}_1\cup \cdots \cup \mathcal{V}_T$ contains all the vertices except the source and the destinations in the graph. For $s=0, 1, \ldots, T$, let $x_s$ be the number of out-vertices in $\overline{K}$.

To obtain the cut-value of an arbitrary $\mathsf{S}-\mathsf{DC}$ cut $K$, we examine the in-edges of all the vertices in $\mathcal{V}_0\cup\mathcal{V}_1\cup \cdots \cup\mathcal{V}_T$, and express the cut-value as a sum of $T+1$ terms:
\begin{equation}
C(K)= \sum_{0\leq s\leq T}C_{\vartriangle,s}(K),
\end{equation}
where $$C_{\vartriangle,s}(K)  \triangleq \sum_{i\in \mathcal{V}_s\cap K, j\in \mathcal{V}_s\cap \overline{K}}u_{ij}$$ is called the \emph{cut-value contribution} of the vertices in $\mathcal{V}_s$. When there is no ambiguity, we may simply write it as $C_{\vartriangle,s}$. For example, in Fig.~\ref{fig:define_example}, the cut denoted by line~2 has cut-value equal to $C_{\vartriangle,0}+C_{\vartriangle,1}+C_{\vartriangle,2} = 0 + \alpha + 3\beta$.

For any $\mathsf{S}$-$\mathsf{DC}$ cut $K$, it is obvious that $\mathsf{DC}\in \overline{K}$. By definition, the $\mathsf{DC}$ has $k$ out-vertices as its parents, and these $k$ edges all have infinite capacity. If the cut-value is finite, then these $k$ out-vertices must be in $\overline{K}$. We can always find $k$ repair rounds, together with the initial stage, which contain all these $k$ out-vertices. For ease of presentation, we re-index these repair rounds as $1, 2, \ldots, k$. We consider only these $k$ (re-indexed) repair rounds to obtain a lower bound of $C(K)$:
\begin{equation}
C(K)\geq C_{\vartriangle,0}+\sum_{1\leq s\leq k}C_{\vartriangle,s}(K). \label{eq:sum_contri}
\end{equation}
From now on, the remaining $T-k$ repair rounds that are not re-indexed will not occur in our discussion. We will consider only the re-indexed repair rounds, with index set $\mK$.

In the initial stage, since there is no auxiliary vertex in repair round~0 and all in-vertices should be in $K$ if $C(K)\neq \infty$, we have
$
C_{\vartriangle,0}=x_0\alpha.
$

For the repair rounds in $\mK$, we have two cases.
First, consider the case where $s\in \mT_1$, where $\mT_1 \triangleq \{s \in \mK : \mathcal{A}_s \cap K \neq \emptyset\}$. In other words, $s \in \mT_1$ if there exists at least one $\mathsf{h}_i^s$ in $K$. We investigate the three classes of vertices in $\mathcal{V}_s$, i.e., auxiliary vertices, in-vertices, and out-vertices, one by one. For the auxiliary vertices, denote the number of $\mathsf{h}_i^s$'s such that it is in $\overline{K}$ and its parent vertex $\mathsf{Out}_i$ is in $K$ by $z_s$. For the in-vertices, we only need to consider the case where all of them are in $K$, for otherwise, the cut-value contribution would be infinite, since all its in-edges have infinite capacity and by definition of $\mT_1$, at least one of its parent vertices (i.e. $\mathsf{h}_i^s$) is in $K$. For the out-vertices, since all in-vertices are in $K$, and the number of them in $\overline{K}$ is $x_s$, we have
$
C_{\vartriangle,s}=x_s\alpha+z_s\beta.
$
An illustration of this case is shown in Fig.~\ref{fig:analysis_T1}. The vertices in the left side of the dash line belong to $K$ (denoted as $\mX'$ in the figure).

Second, consider the case where $s\in \mT_2\triangleq \mK\setminus\mT_1$. By definition of $\mT_2$, all $\mathsf{h}_i^s$'s are in $\overline{K}$.
For the auxiliary vertices, denote the number of $\mathsf{h}_i^s$'s such that its parent vertex $\mathsf{Out}_i$ is in $K$ by
$y_s$. For all the in-vertices, since their parent vertices (i.e. $\mathsf{h}_i^s$) are all in $\overline{K}$, their cut-value contribution is zero, no matter they are in $K$ or $\overline{K}$. For the out-vertices, denote the number of them such that $\mathsf{Out}_j\in \overline{K}$ and its parent vertex $\mathsf{In}_j\in K$ by $v_s$.
we have
$
C_{\vartriangle,s}=v_s\alpha+y_s\beta.
$
An illustration of this case is shown in Fig.~\ref{fig:analysis_T2}. The vertices in the left side of the dash line belong to $K$ (denoted as $\mX'$ in the figure).

Combining the above two cases and according to~\eqref{eq:sum_contri}, we have
\begin{align}
C(K)\geq x_0\alpha+\sum_{s\in \mT_1}\left(x_s\alpha+z_s\beta\right)+\sum_{s\in \mT_2}(v_s\alpha+y_s\beta)\notag.
\end{align}

Now consider a special cut $K^*$, which is constructed from $K$ as follows. Initially, let $K^*$ be the same as $K$. For $s\in \mT_1$, move all $\mathsf{h}_i^s$'s into $K^*$, and then $z_s$ becomes zero.
Note that, since $\mathsf{h}_i^s$'s child vertices are all in round $s$, moving all $\mathsf{h}_i$'s into $K^*$ will not affect the cut-value contribution of other repair rounds. For $s\in \mT_2$, move all $\mathsf{In}_j$'s into $\overline{K^*}$, and $v_s$ becomes zero. Again, since $\mathsf{In}_j$'s child vertex $\mathsf{Out}_j$ is in the same repair round, moving $\mathsf{In}_j$ will not affect the cut-value contribution of other repair rounds. For the example in Fig.~\ref{fig:analysis_T1} and Fig.~\ref{fig:analysis_T2}, the corresponding $K^*$ (denoted as $\mX$ in the figures) are the vertices in the left side of the solid line.
We have
\begin{align}
C(K)\geq C(K^*)\geq x_0\alpha+\sum_{s\in \mT_1}x_s\alpha+\sum_{s\in \mT_2}y_s\beta,\label{eq:bound_x'}
\end{align}
where $\mT_1$ is the index set of repair rounds whose auxiliary vertices and in-vertices are all in $K^*$, and $\mT_2$ is the index set of repair rounds whose auxiliary vertices and in-vertices are all in $\overline{K^*}$.

The newcomers in round~$s$ are connected to $d$ helpers, which are located in rounds $0, 1, \ldots, s-1$. Of these $d$ helpers, at most $\sum_{i=0}^{s-1}x_i$ have their out-vertices in $\overline{K^*}$. Therefore, we have
\begin{align*}
y_s \geq d-\sum_{i=0}^{s-1}x_i,\text { for } s\in \mT_2.
\end{align*}
Thus we obtain the lower bound~\eqref{eq:bound}.

Next, we consider the constraints. It is clear that~\eqref{eq:x0} and~\eqref{eq:xs_range} must hold. Since the initial stage and the $k$ repair rounds have $k$ out-vertices in $\overline{K}$, we must have
\begin{equation}
x_0 + \sum_{s\in\mK}x_s\geq k. \label{eq:sum_xs_more_k}
\end{equation}
Now we show that the inequality in~\eqref{eq:sum_xs_more_k} can be replaced by an equality. To see this, suppose $(x_0^*, x_1^*, \ldots, x_k^*, \mT_1^*)$ is an optimal solution, which yields the minimum value $B^*$. Let $l$ be the first repair round after which $\sum_{s=0}^l x_s^* \geq k$. Suppose $l$ is not the last round (i.e. $l \neq k$). If $k \in \mT_1^*$, then $x_k^*$ must be equal to 0, for otherwise, we can reduce its value and $B^*$ cannot be the minimum. If $k \in \mT_2^*$, we move $k$ into $\mT_1^*$ and set $x_k^* := 0$. Since this does not change the value of $B^*$, we can assume $k \in \mT_1^*$ and $x_k^* = 0$. The same argument can be repeatedly applied to round $k-1$, round $k-2$ and so on, so that we can assume $s \in \mT_1^*$ and $x_s^* = 0$ for all $s > l$. Consider round $l$. If $l \in \mT_1^*$, then $\sum_{s=0}^l x_s^*$ must be equal to $k$, for otherwise we can reduce the value of $x_l^*$ to obtain a value lower than $B^*$. If $l \in \mT_2^*$, we can reduce the value of $x_l^*$ by $\sum_{s=0}^l x_s^*-k$ without changing the value of $B^*$. Hence, replacing the inequality in~\eqref{eq:sum_xs_more_k} by an equality does not affect the value of the lower bound.
\end{proof}

\begin{myth} \label{th:mincut}
$C_{\text{storage}}^T = C_{\text{storage}}^\infty = B$ for all $T \geq k$.
\end{myth}

\begin{proof}
Theorem~\ref{th:bound} states that $C_{\text{storage}}^T \geq B$ for all $T \geq k$. It remains to prove that the lower bound is tight. Let $(\bx^*,\mT_1^*)$ be an optimal solution to the minimization in Theorem~\ref{th:bound}. We construct an instance $I^*$ with a particular failure pattern, $\mathsf{DC}^*$ and a cut $\mX^*$ such that the cut-value $C(\mX^*)$ is exactly $B$. 

The instance $I^*$ is constructed as follows. First, we specify the failure pattern. In stage 0, choose any $r$ nodes in $\mathcal{N}$ and let them fail. For stage $s\in\mK$, choose any $r-x_s^*$ nodes in $\mR_s$ and any $x_s^*$ active nodes in $\mathcal{N}$ and let them fail right before stage $s+1$. We can always find such a failure pattern since there are $r$ nodes in $\mR_s$ for every $s$, and the accumulated number of failed nodes in $\mathcal{N}$ is
$$
r+\sum_{s\in\mK}x_s^*\leq r+k\leq n,
$$
where the first inequality follows from~\eqref{eq:sum_xs} and the second inequality follows from the assumption in the system model. Since the remaining active nodes in $\mathcal{N}$ is
$$
n-(r+\sum_{s\in\mK}x_s^*)=n-(r+k-x_0^*)\geq x_0^*,
$$
we can select any $x_0^*$ active nodes from $\mathcal{N}$ and denote them by $\mathcal{M}_0$. Denote the active nodes in $\mR_s$ by $\mathcal{M}_s$. 

Next, we specify the helpers for each repair round. The helpers for repair round $i$, for $i = 1, 2, \ldots, s$, are chosen first from $\mathcal{M}_0$, then from $\mathcal{M}_1$, and so on, until $d$ helpers are chosen. If $\sum_{i=0}^{s-1}|\mathcal{M}_i|<d$, the remaining helpers are chosen arbitrarily from the active nodes in $\mathcal{N}$. There are always enough active nodes in $\mathcal{N}$ to serve as helpers because
\begin{align}
&n-r-\sum_{i\in\mK, i<s }x^*_i\label{eq:valid1}\\
&= n-r-|\mathcal{M}_1|-\dots-|\mathcal{M}_{s-1}|\label{eq:valid2}\\
&\geq d-|\mathcal{M}_0|-\dots-|\mathcal{M}_{s-1}|,\label{eq:valid3}
\end{align}
where~\eqref{eq:valid1} is the number of active nodes in $\mathcal{N}$ after stage $s-1$, ~\eqref{eq:valid2} follows from the definition of $\mathcal{M}_s$, and~\eqref{eq:valid3} is the number of required helpers in $\mathcal{N}$. The inequality holds because $n-r\geq d$ according to the system model.


Finally, consider $\mathsf{DC}$, which comes after the
repair round $k$ and connects to $\mathcal{M}_0 \cup \mathcal{M}_1 \cup \dots \cup \mathcal{M}_k$.

The cut $\mX^*$ is constructed as follows: For $s\in \{0\}\cup \mT_1^*$, put $\mathsf{Out}_i$, for $i\in\mathcal{M}_s$, into $\overline{\mX^*}$, and all the remaining vertices in round~$s$ into $\mX^*$.
Vertices in these repair rounds contribute $x_0^*\alpha+\sum_{s\in \mT_{1}^*}x_s^*\alpha$ to the cut-value.
For $s\in \mT_2^*=\mK\setminus\mT_1^*$, put all vertices in round~$s$ into $\overline{\mX^*}$. Vertices in these repair rounds contribute
$$\sum_{s\in \mT_{2}^*}(d-x_0^*-\sum_{i\in \mK, i < s}x_i^*)\beta$$
to the cut-value.
Summing up the cut-value contribution of all the vertices, we get $B$, 
showing that the bound is tight.
\end{proof}



We give an example to demonstrate the above result. Consider the storage system WDSS$(8,3,4,2,\alpha,\beta,2)$, with an optimal solution to the minimization in Theorem~\ref{th:bound} as $\mT_1^*=\{1,3\}$ and $x^*_0=0$, $x^*_1=1$, $x^*_2=2$, $x^*_3 =0$, which implies that $\mT_2^*=\{2\}$ and $|\mathcal{M}_0| = x_0^* = 0,  |\mathcal{M}_1| = x_1^* = 1, |\mathcal{M}_2| = x_2^* = 2, |\mathcal{M}_3| = x_3^* = 0$. Accordingly, the lower bound is
\begin{align}
B &= x_1^* \alpha + (d-|\mathcal{M}_0|-|\mathcal{M}_1|) \beta \nonumber \\
&=\alpha+3\beta.\label{eq:achieve_exam}
\end{align}
An instance is indicated in Fig.~\ref{fig:define_example}, where nodes~5 and~6 failed in round~1 and nodes~8 and~10 failed in round 2.
Let $\mathcal{M}_0=\emptyset$, $\mathcal{M}_1=\{9\}$, and $\mathcal{M}_2=\{11,12\}$. Let $\mathsf{DC}$ be connected to $\{9,11,12\}$. It can be seen that the cut-value of Cut~2, indicated in Fig.~\ref{fig:define_example}, meets the lower bound $B$ in~\eqref{eq:achieve_exam}.

\begin{figure*}
 \centering
 \includegraphics[width=6in]{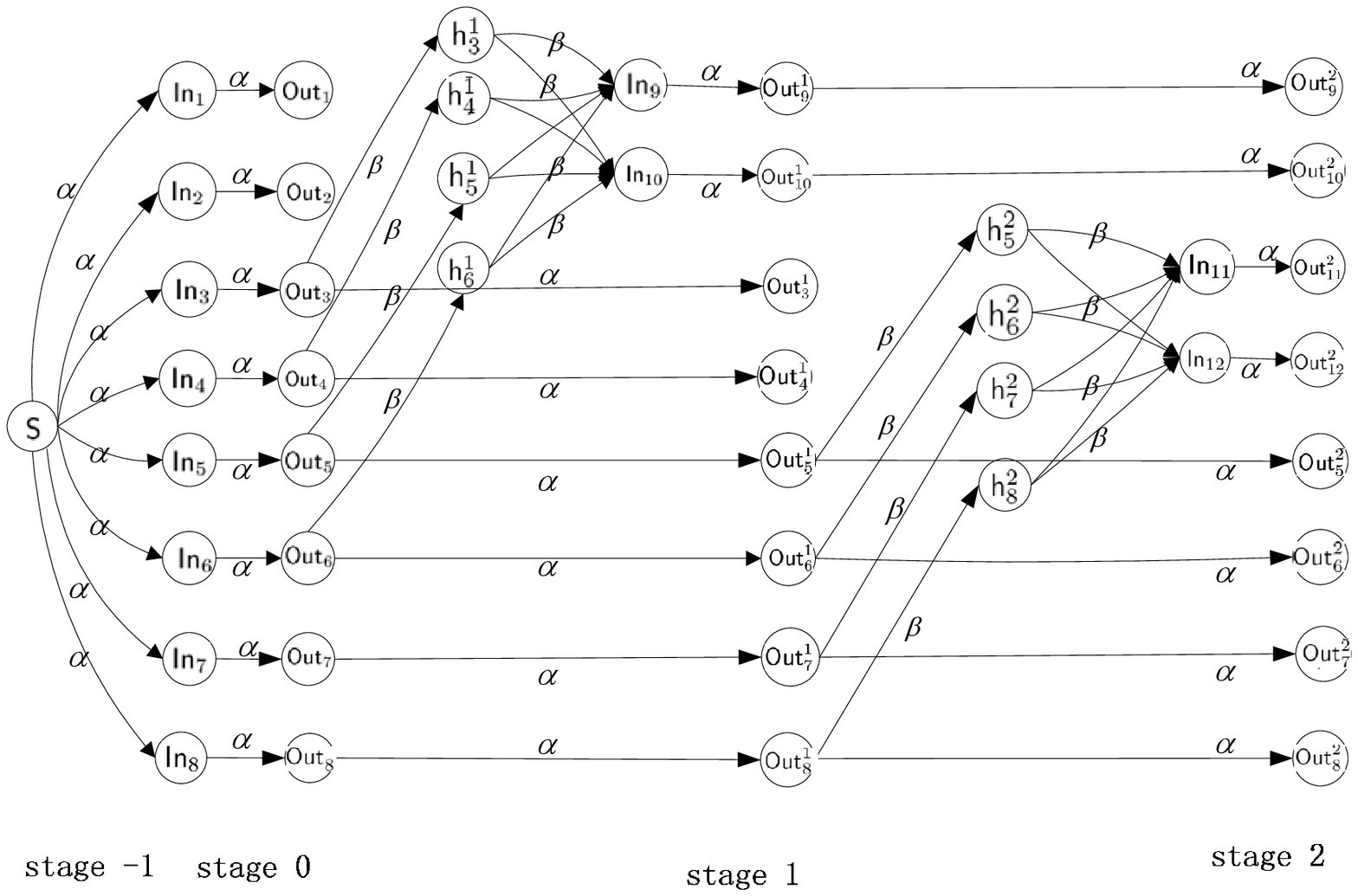}
 \caption{An example of a refined information flow graph. Each link with label $c$ represents $c$ parallel edges of unit capacity.}
 \label{fig:refined}
\end{figure*}

\section{Achievability of the Min-Cut Bound}\label{sec:inf_rounds}

In this section, we show that for a WDSS with infinite repair rounds, the min-cut bound can be achieved by linear network coding with a finite alphabet. In other words, we show that $C_{\text{storage}}=B$. To do this, we follow the idea of generic storage codes~\cite{itw15}, which is based on the concept of generic network coding~\cite{Li,Yeung_book}. In the following, we first introduce the refined information flow graph of the WDSS. Then we present the concept of generic storage codes.
Finally, the achievability is shown by the existence of a generic storage code for this particular network.

\subsection{Refined Information Flow Graph}

Given any information flow graph for a WDSS, we construct a
\emph{refined information flow graph} by introducing the concept
of {\em repair stage}. In regard to the refined information flow
graph, the repair process of $r$ nodes is called a repair stage. From
$\mathsf{S}$ to $\mathsf{In}_i$'s is called stage $-1$. The
original $n$ storage nodes are said to be in stage $0$. In stage
$s>0$, the out-vertex of each storage node, except the failed ones
in stage $s-1$, is connected to an auxiliary out-vertex by a
directed edge of capacity $\alpha$. This is a distinctive feature of
the refined graph, which differs from the original flow graph. In addition,
we add stage stamps for the out-vertices and auxiliary out-vertices. The out-vertex of node~$i$ in
stage $s$ is denoted by $\mathsf{Out}_i^{s}$. For any stage $s$, if node $i$ is a helper, then a new auxiliary h-vertex $\mathsf{h}_i^{s}$ is added to $\mathsf{Out}_i^s$ by an edge with capacity $\beta$.
Newcomers in this stage is connected to these new auxiliary h-vertices. This construction is essentially the same as the original flow graph.
For the out-vertices of the newcomers, we also need to add stage stamps, i.e., re-label the out-vertex of the newcomer $t$ as $\mathsf{Out}_t^{s}$.

After constructing the vertices and connecting them by edges as described above, we adjust the capacities of some edges in a way which does not affect the capacity of the network. Since each storage node has capacity $\alpha$, every edge from
$\mathsf{S}$ to $\mathsf{In}_i$ with infinite capacity is replaced
by an edge of capacity $\alpha$. Furthermore, since each auxiliary h-vertex's in-edge only has capacity $\beta$, we change all its out-edges' capacity from infinite to $\beta$. After these changes, all the edge capacities are finite. Finally, for each edge in the network, if its capacity is $c$, we replace it by $c$ parallel edges of unit capacity.

Note that the refined information flow graph represents a single-source multicast
acyclic network. An example with $n=8, d=4$, and $r=2$ is shown in
Figure~\ref{fig:refined}. In this example, nodes~$1$ and $2$ fail in
stage~0 and nodes $3$ and $4$ fail in stage~$1$.

\subsection{Generic Storage Codes}
Based on the refined information flow graph, we can introduce the generic storage codes of the WDSS. Before that, we first give some concepts of network coding~\cite{Yeung_book}, which are building blocks of generic storage codes.

Consider a single-source acyclic communication network and its
corresponding graph, $G=(V,E)$. Let the alphabet $\Sigma$ be the
finite field $GF(q)$. 
Suppose the message to be transmitted from the source node is an
$\omega$-dimensional column vector $\bx$ over $GF(q)$. We add
$\omega$ imaginary edges terminating at $\mathsf{S}$ and assign
them distinct vectors of the
$\omega$-dimensional standard basis. These vectors are referred to
as the global encoding kernels of the imaginary edges. For each
edge $e (i,j) \in E$, we iteratively define its global encoding
kernel by
$$
\bg_e\triangleq \sum_{d\in\text{I}(i)}l_{d,e} \bg_d,
$$
where $\text{I}(i)$ is the set of all incoming edges of $i$.
The transmitted symbol on edge $e$ is $\bx^T \bg_e$. An
$\omega$-dimensional linear network code consists of a scalar $l_{d,e}\in GF(q)$ for every adjacent pair of edges $<d,e>$ in the network
as well as a column $\omega$-vector $\bg_e$ for every edge $e$.

For an edge
set $P$, denote the set of the corresponding global encoding
kernels by
$$
\text{ker}(P) \triangleq \{\bg_e : e\in P\},
$$
and the linear span of $\text{ker}(P)$ by
$$
\text{vspace}(P) \triangleq \text{span}(\text{ker}(P)).
$$
For a vertex $i$, define
$$
\text{vspace}(i) \triangleq \text{span}(\text{ker}(\text{I}(i))).
$$


A sequence of edges $e_1, e_2, \ldots, e_n$ forms a {\em path} if
Head($e_i$) = Tail($e_{i+1}$) for $1\leq i \leq n-1$. Two paths
are {\em edge-disjoint} if they do not have any edge in common. A
set of edges is said to be {\em path-independent} if each edge in this
set is on a path originating from an imaginary edge and these
paths are edge-disjoint. An edge set $P$ is said to be {\em
regular} with respect to a linear network code if the global
encoding kernels in $\text{ker}(P)$ are linearly independent.

Since a $\mathsf{DC}$ can only connect to nodes in the same stage in the refined information flow graph, to make sure all $\mathsf{DC}$ can retrieve the file, it is sufficient to ensure all path-independent sets of  edges in the same stage are regular. A code that satisfies
this requirement can be regarded as a restricted form of a generic network code~\cite{Li,generic}. We call it a {\em generic storage code}. Denote the set of all the edges in stage $s$, except the incoming edges of data collectors, by $E_s$. Generic storage codes can then be formally defined below:

\begin{mydef}
  An $\omega$-dimensional linear network code on a refined information flow graph
  is said to be an {\em $\omega$-dimensional generic storage code} if every path-independent $\omega$-subsets of $E_s$ is regular,
  for any stage $s = 0, 1, 2, \ldots$.
\end{mydef}

\subsection{Achievablity}

Theorem~\ref{th:mincut} says that the min-cut value of the network is $B$. The question is whether this min-cut value, which is an upper bound of the storage capacity, is achievable. Note that the classical multicast network coding result does not apply, since the graph is infinite. In the following, we show that $B$ is achievable by the use of generic storage codes.
\begin{myth}
A file with size $\omega = B$ can be stored in a WDSS($n,k,d,r,\alpha,\beta,\infty$) by the use of an
$\omega$-dimensional generic storage code over GF($q$), where $q> {n\alpha+d\beta \choose \omega-1}$.
\label{main}
\end{myth}

\begin{proof}

If $q> {n\alpha+d\beta \choose \omega-1}$, we can construct an $\omega$-dimensional generic storage code over GF($q$) on a refined information flow graph of a WDSS. The existence of such a code is shown in Lemma~\ref{le:generic} in Appendix~\ref{app:constrcut code}. Now we show that any data collector can retrieve the file based on the code.
Recall that any data collector is connected to $k$ out-vertices in the same stage.
By Theorem~\ref{th:mincut}, the value of a cut is at least $B$, so there are at least $B$ disjoint paths terminating at any $k$ out-vertices, in every stage $s\geq 0$. Thus there are at least one path-independent set, say $P$, with size $B$ within the incoming edges of these $k$ out-vertices. By the definition of generic storage codes,
the dimension of ker($P$) is $B$, and the file with size $B$ can be decoded.
\end{proof}

The generic storage code can be constructed by Algorithm 1 in  Appendix~\ref{app:constrcut code}. In each stage, there are at most $n\alpha+d\beta \choose \omega-1$ subsets to be considered for assigning a global encoding kernel for each edge.
Since the total number of edges to be processed in each stage is $d\beta+r\alpha$, for a given value of $\omega$, the complexity of the algorithm is polynomial time in $n\alpha+d\beta$ in each stage.

\section{Comparison with Cooperative Repair}\label{sec:compare}

We compare broadcast repair with cooperative repair, assuming both repair processes are triggered after the number of failed storage nodes accumulates to $r$. To simplify the analysis, in the following, we restrict $r$ to be a divisor of $k$, i.e.,
\begin{align}
ru=k, \label{eq:u_def}
\end{align}
for some positive integer $u$. The storage capacity of a WDSS can then be expressed as follows:
\begin{myth}\label{thm:capacity_sol_form}
If~\eqref{eq:u_def} holds, the storage capacity of a WDSS is
  \begin{equation}
    C_{\text{storage}}=\sum_{j=1}^{u}\min\left\{r\alpha,(d-(j-1)r)\beta\right\}.\label{eq:capacity_sol_form}
  \end{equation}
\end{myth}
\begin{proof}
Define
\begin{equation} \label{eq:c_mT}
c_{\mT_1}(\bx) \triangleq x_0\alpha+\sum_{s\in\mT_1}x_s\alpha+\sum_{s\in\mT_2}(d-\sum_{i=0}^{s-1}x_i)\beta,
\end{equation}
where $\bx=(x_0,x_1,\dots,x_k)$ satisfying~\eqref{eq:x0} to~\eqref{eq:sum_xs}.
Based on Lemma~\ref{th:bound_sol} shown in Appendix~\ref{app:no positive}, there is an optimal solution vector $\bx^*=(x_0^*,x_1^*,\dots,x^*_{k})$ for $\min c_{\mT_1}$, whose components are either $x^*_i=r$ or $x^*_i=0$. Because of~\eqref{eq:sum_xs} and~\eqref{eq:u_def}, there are $u$ components that have the value $r$.

Recall that $C_{\text{storage}} = \min_{\mT_1} \min_{\bx} c_{\mT_1}(\bx)$, where $\mT_1 \subseteq \mathcal{K}$ and $\bx$ is subject to~\eqref{eq:x0} to~\eqref{eq:sum_xs}. Let $\mT_1^*$ be an optimal solution and denote $|\mT_1^*|$ by $t_1$. It can be seen from~\eqref{eq:c_mT} that there exists an optimal $\bx^*$ in the form of
$$
(\underbrace{r, r, \ldots, r}_{t_1}, \underbrace{r, r, \ldots, r}_{u-t_1}, \underbrace{0, 0, \ldots, 0}_{k-u}).
$$
Hence,
$$
C_{\text{storage}} = \min_{t_1\in \{1, \ldots, u\}} t_1 r\alpha + \sum_{j=t_1+1}^u (d- (j-1)r) \beta,
$$
which can be re-written as~\eqref{eq:capacity_sol_form}.
\end{proof}


Now we compare broadcast repair with cooperative repair. In cooperative repair, the newcomers receive packets from helper
nodes through individual channels, and then exchange the encoded
packets to all the other newcomers.

Consider the two points, \emph{minimum storage} (MS) point, which corresponds to the best storage efficiency, and the \emph{minimum repair-transmission bandwidth} (MT) point, which corresponds to the minimum repair-transmission bandwidth on the tradeoff curve between repair-transmission bandwidth and storage (see Fig.~\ref{fig:compare} for example). In cooperative repair, the repair-transmission bandwidth is equal to the repair bandwidth. According to~\cite{kenneth_jnl}, the MS point and the MT point for cooperative repair are
$$
(\tau_{\text{MSC}},\alpha_{\text{MSC}})=(\frac{d+r-1}{k(d+r-k)},\frac{1}{k}),
$$
and
$$
(\tau_{\text{MTC}},\alpha_{\text{MTC}})=\frac{2d+r-1}{k(2d+r-k)}(1,1),
$$
respectively. Note that the values in the above expressions are normalized so that the file size (or equivalently, the storage capacity) is normalized to 1.

Next we consider broadcast repair.
To derive the MS point and MT point for broadcast repair,
for the purpose of normalization, $C_{\text{storage}}$ is assumed to be 1. At the MS point, $\alpha$ is equal to $1/k$ so that $k$ nodes can recover the file. When $\alpha=1/k$, every term of~\eqref{eq:capacity_sol_form} should be $r\alpha$ such that the sum of $u$ terms should equal 1.
It is required that
 $$
 (d-(u-1)r)\beta \geq r\alpha.
 $$
Then we can obtain the MS point for broadcast repair as follows:
$$
(\tau_{\text{MSB}},\alpha_{\text{MSB}})=\big(\frac{d}{k(d+r-k)},\frac{1}{k}\big).
$$
At the MT point, the total number of transmitted symbols is equal to the total number of stored symbols in the newcomers. Therefore, we have $\alpha = \tau$. Hence, $r\alpha=d \beta\geq (d-(j-1)r)\beta$ for $j=1,2,\dots,u$. According to~\eqref{eq:capacity_sol_form}, we have
\begin{align*}
\sum_{j=1}^u (d-(j-1)r)\beta &=1 \\
u \big[ d - \frac{r(u-1)}{2} \big] \beta &= 1,
\end{align*}
from which we can obtain the MT point for broadcast repair as follows:
$$
(\tau_{\text{MTB}},\alpha_{\text{MTB}})=\frac{2d}{k(2d+r-k)}(1,1).
$$

It is easy to see that broadcast repair outperforms cooperative
repair in the two points for any $r > 1$.
In Fig.~\ref{fig:compare}, we plot the tradeoff curves of the two repair schemes with
parameters $C_{\text{storage}}=1, d=9$, $k=4$ and $r=2$. We have $$
(\tau_{\text{MSB}},\alpha_{\text{MSB}})=\big(0.321,0.25\big),
$$
and
$$
(\tau_{\text{MTB}},\alpha_{\text{MTB}})=(0.281,0.281).
$$
As a benchmark, we also plot
the single-node repair, in which the repair is triggered whenever there is
a single node failure.
As reported in~\cite{kenneth_jnl}, cooperative repair performs
better than single-node repair due to the benefit of node cooperation.
On the other hand, when applied to WDSS, it performs worse than broadcast repair, since it does not exploit
the broadcast nature of the wireless medium.

\begin{figure}
 \centering
 \includegraphics[width=3.1in,angle=0]{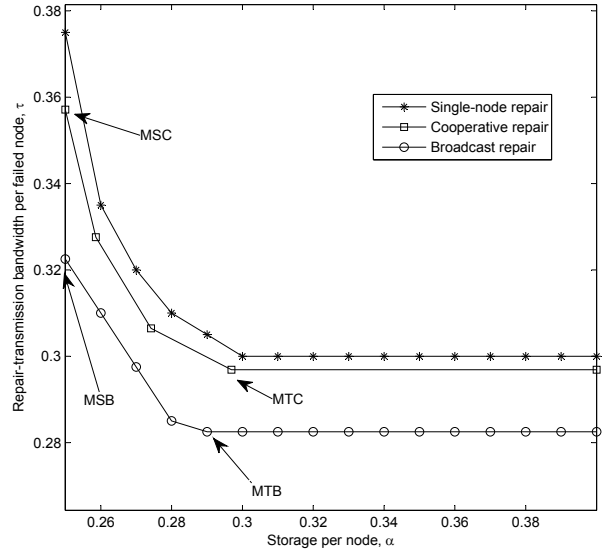}
 \caption{Tradeoff curve between repair-transmission bandwidth and storage, $C_{\text{storage}}=1,d=9,k=4,r=2$.}\label{fig:compare}
\end{figure}

%

\section{Conclusions}\label{sec:conclu}

In this work, we show that by exploiting the broadcast nature of the wireless channel,
the performance of a WDSS can be improved. Based on the graph representation of the WDSS, we derive the
storage capacity of the WDSS. The fundamental tradeoff between the repair-transmission bandwidth and storage amount shows that
broadcast repair for $r$ nodes outperforms the naive way of repairing these $r$ nodes one by one. We also compare broadcast repair with
cooperative repair, both of which are specifically designed for repairing multiple nodes. While cooperative repair works very well
in wired DSS, it is outweighed by broadcast repair in wireless environments.

We have shown that the optimal tradeoff under functional repair can be achieved by generic storage codes. Except for some special cases, it remains unknown whether exact-repair codes exist on the whole curve. It is theoretically challenging and practically important to construct such codes. We hope that our work can stimulate more studies on this interesting topic.

In this work, our model assumes that all the wireless storage nodes are within radio coverage of one another. In reality, the network topology may not be fully connected, and even if it is fully connected, the channel gains of different links can be different. Our model, however, provides a baseline study for performance evaluation of systems with other network topologies and more sophisticated physical-layer techniques of using the broadcast channel. Finally, we remark that our model is not limited only to WDSS. As it is equivalent to the centralized repair model considered in~\cite{centralized_repair, centralized_repair_j}, our results can be directly applied to a DSS which performs repairs at a central location. This, for example, includes the rack-based architecture, under which when there is a rack failure, all the nodes within the rack need to be replaced and centralized repair can be performed at a leader node in the rack.

\bibliographystyle{IEEEtran}

\bibliography{zigzag}

\appendices
\section{}\label{app:constrcut code}
\begin{myle} \label{le:generic} Let $G$ be a refined
information flow graph of a WDSS such that there is at least one path-independent set of edges of size $\omega$ in each stage. An $\omega$-dimensional generic storage code on $G$ over GF($q$) can be constructed, provided that $q>{n\alpha+d\beta \choose \omega-1}$.
\end{myle}

\begin{proof}
Let $q$ be a prime power greater than ${n\alpha+d\beta \choose
\omega-1}$. We prove the statement by mathematical
induction on the number of repair stages in the refined
information flow graph. We want to maintain the inductive invariant that, in any stage, any path-independent set of edges is regular.

Consider a refined information graph in stage $-1$. First, note that any $\omega$-subset of the edges in stage $-1$ is path-independent. We claim that there exists a linear
code such that all these $\omega$-subsets are regular. For the first $\omega$
edges in stage $-1$, it is clear that they can be assigned linearly independent global encoding kernels.
For each of the subsequent edges in stage $-1$, we can pick a vector $\bx \not \in \text{vspace}(\zeta)$, where
$\zeta$ is any $(\omega-1)$-subset of edges that have already been assigned global encoding kernels. This can be done by picking a generator matrix of an $\omega$-dimensional Reed-Solomon code of length $n\alpha$. We can also assign the global encoding kernels sequentially, since
$$
\big|\bigcup_\zeta \text{vspace}(\zeta)\big| \leq {n\alpha \choose \omega-1} q^{\omega-1} < q^\omega.
$$

In stage $0$, let the global encoding kernels of $\alpha$
outgoing edges of every $\mathsf{In}_i$, $i=1,2,\dots,n$ be the same
as those of its $\alpha$ incoming edges. Since in stage $-1$, any
$\omega$-subset of the $n\alpha$ edges is regular, so is any $\omega$-subset
of the $n\alpha$ edges in stage~0.

Assume that for any refined information graph with $s\geq 0$ repair stages, (i.e., $s+1$ stages
including stage 0), a generic storage code has been constructed. By definition, any path-independent
$\omega$-subset of $E_s$ is regular with respect to the constructed
network code.

In stage $s+1$, there are $n-r$ auxiliary out-vertices of the surviving nodes
in stage $s$, $d$ new auxiliary h-vertices of the helpers, and $r$ newcomers. Let the set of indices of the $n-r$ surviving nodes be $A$.
For $i\in A$, $\mathsf{Out}_i^{s}$ has $\alpha$ incoming edges and
$\alpha$ outgoing edges connecting to $\mathsf{Out}_i^{s+1}$. Let the global encoding kernels of
these $\alpha$ outgoing edges be the same as those of the $\alpha$ incoming edges.
Let the $d$ helpers be indexed by $h_1, h_2, \ldots, h_d \in A$, where $h_1 < h_2 < \cdots < h_d$, and
the index of the newcomers be~$t,t+1,\dots,t+r-1$. It remains to determine the global encoding kernels for all the incoming edges $\mathsf{h}_i^{s+1}$, $i=1,2,\dots,d$, $\mathsf{In}_t,\mathsf{In}_{t+1},\dots,\mathsf{In}_{t+r-1}$, and $\mathsf{Out}_t^{s+1},\mathsf{Out}_{t+1}^{s+1},\dots,\mathsf{Out}_{t+r-1}^{s+1}$ in such a way that any path-independent $\omega$-subset of $E_{s+1}$ is regular. This can be done by Algorithm~\ref{generic_code}, which is adapted from~\cite[Algorithm 19.34]{Yeung_book}. While the algorithm in~\cite[Algorithm 19.34]{Yeung_book} considers all edges in a graph, Algorithm 1 only needs to consider the edges within the same stage. Besides, in Algorithm 1, the global encoding kernels of some edges are directly obtained from previous stages, which is different from the algorithm in~\cite[Algorithm 19.34]{Yeung_book}.

\renewcommand{\algorithmicrequire}{\textbf{Input:}}
\renewcommand{\algorithmicensure}{\textbf{Output:}}

\begin{algorithm}
\caption{Assign Global Encoding Kernels for the New Auxiliary h-vertices and Newcomers} \label{generic_code}
\begin{algorithmic}[1]

\REQUIRE $\{\bg_e : e \in \text{I}(\mathsf{Out}_i^{s}) \text{ for all }i \in A \}$ and $h_1, h_2, \ldots, h_d$

\ENSURE $\{\bg_e : e \in \text{I}(\mathsf{h}_i^{s+1}),i=1,2,\dots,d,$ \\ or $ e \in \text{I}(\mathsf{In}_t),\text{I}(\mathsf{In}_{t+1}),\dots,,\text{I}(\mathsf{In}_{t+r-1}),$ \\ or $ e \in \text{I}(\mathsf{Out}_t^{s+1}),\text{I}(\mathsf{Out}_{t+1}^{s+1}),\text{I}(\mathsf{Out}_{t+r-1}^{s+1})\}$

\STATE $B_0 := \{e \in \text{I}(\mathsf{Out}_i^{s}) \text{ for all }i \in A \}$;

\FOR {$i :=h_1, h_2, \ldots h_d$}

\FOR {$j :=1, 2, \ldots \beta$}

\STATE $e := $ the $j$-th incoming edges of $\mathsf{h}_i^{s+1}$ from $\mathsf{Out}_{h_i}^{s}$;

\STATE Choose a vector $\bx \in \text{vspace}(\mathsf{Out}_{h_i}^{s})$ such that $\bx \not \in \text{vspace}(\zeta)$, where $\zeta$
is any $\omega-1$-subset of $B_0$ such that $\zeta$ is regular and $\text{vspace}(\mathsf{Out}_{h_i}^{s}) \not \subset \text{vspace}(\zeta)$;

\STATE $\bg_e := \bx$ and $B_0 := B_0 \cup \{e\}$;

\ENDFOR

\ENDFOR

\FOR {$i :=h_1, h_2, \ldots h_d$}
\FOR {$j :=1,2,\ldots \beta$}
\STATE $e := $ the $j$-th incoming edges of $\mathsf{h}_i^{s+1}$ from $\mathsf{Out}_{h_i}^{s}$;
\FOR {$l :=t,t+1, \ldots t+r-1$}

\STATE $e^{\prime} := $ the $j$-th outgoing edges of $\mathsf{h}_i^{s+1}$ to $\mathsf{In}_{l}$;
\STATE $\bg_{e^{\prime}} := \bg_e$ ;

\ENDFOR
\ENDFOR
\ENDFOR

\FOR {$i :=t,t+1,\ldots t+r-1$}
\FOR {$j :=1, 2, \ldots \alpha$}

\STATE $e := $ the $j$-th incoming edges of $\mathsf{Out}_i^{s+1}$;

\STATE Choose a vector $\bx \in \text{vspace}(\mathsf{In}_i)$ such that $\bx \not \in \text{vspace}(\zeta)$, where $\zeta$
is any $\omega-1$-subset of $B_0$ such that $\zeta$ is regular and $\text{vspace}(\mathsf{In}_i) \not \subset \text{vspace}(\zeta)$;

\STATE $\bg_e := \bx$ and $B_0 := B_0 \cup \{e\}$;

\ENDFOR
\ENDFOR

\end{algorithmic}
\end{algorithm}

By construction, it can be seen that any path-independent $\omega$-subset of $E_{s+1}$ is regular. In the algorithm, the vector $\bx$ in line~5 can always be found. To see this, notice that there are at most $(n\alpha+d\beta)$ edges in $B_0$ and the number of possible choices of $\zeta$ is at most ${n\alpha+d\beta \choose \omega-1}$. Denote the dimension of
vspace($\mathsf{Out}_{h_i}^{s}$) by $\nu$. Since $\text{vspace}(\mathsf{Out}_{h_i}^{s}) \not \subset \text{vspace}(\zeta)$, the dimension of $\text{vspace}(\mathsf{Out}_{h_i}^{s}) \cap \text{vspace}(\zeta)$ is less than or equal to $\nu-1$. Thus,
\begin{align*}
\Big|\text{vspace}(\mathsf{Out}_{h_i}^{s}) \cap (\bigcup_\zeta \text{vspace}(\zeta))\Big| &\leq {n\alpha+d\beta \choose \omega-1} q^{\nu-1} \\
&< q^\nu = \left|\text{vspace}(\mathsf{Out}_{h_i}^{s})\right|.
\end{align*}
Likewise, the vector $\bx$ in line~21 can also be found.
\end{proof}

\section{}\label{app:no positive}
\begin{myle}\label{th:bound_sol}
If $ru=k$ for some positive integer value $u$, the problem $\min c_{\mT_1}$ has an optimal solution vector, whose components are either 0 or $r$.
\end{myle}

\begin{proof}
Let $\bx \triangleq (x_0, x_1, \ldots, x_k)$ be an optimal vector, which must exist since there are
only finite possible choices of $\bx$. Since $r$ is a divisor of $k$ and the sum of all $x_i$'s is equal to $k$, $\bx$ cannot have exactly one component whose value is positive and strictly less than $r$. Assume $\bx$ has two or more components which are positive and strictly less than $r$, and we show that we can reduce the number of such components without increasing the value of the objective function.

Suppose $0 < x_l < r$ and $0 < x_j < r$ for some $l, j$ such
that $0 \leq  l<j \leq k$. Denote $\delta$ as any integer which satisfies
\begin{equation}
0<\delta \leq \min \{r - x_l, x_l, r-x_j, x_j\}.\notag
\end{equation}
Let $\bx'$ be the same as $\bx$ except that its $l$-th and
$j$-th components are different from those in $\bx$.
We have
\begin{align*}
&c_{\mT_1}(\bx)= \sum_{s\in\{0\}\cup\mT_1}x_s\alpha+\sum_{s\in\mT_2 \atop s\leq l}(d-\sum_{i=0}^{s-1}x_i)\beta
+\sum_{s\in\mT_2 \atop l<s\leq j}(d-\sum_{i=0}^{s-1}x_i)\beta+\sum_{s\in\mT_2 \atop j<s}(d-\sum_{i=0}^{s-1}x_i)\beta.
\end{align*}
and
\begin{align*}
&c_{\mT_1}(\bx^{\prime})=\sum_{s\in\{0\}\cup\mT_1}x_s'\alpha+\sum_{s\in\mT_2 \atop s\leq l}(d-\sum_{i=0}^{s-1}x_i')\beta
+\sum_{s\in\mT_2 \atop l<s\leq j}(d-\sum_{i=0}^{s-1}x_i')\beta+\sum_{s\in\mT_2 \atop j<s}(d-\sum_{i=0}^{s-1}x_i')\beta.
\end{align*}

In the following three cases: (i) $l, j \in \{0\}\cup \mT_1$, (ii) $l, j \in \mT_2$, or (iii) $l\in \mT_2, j \in \{0\}\cup \mT_1$, let $x_l'=x_l + \delta$ and $x_j'=x_j -\delta$.
Since in these cases $\sum_{s\in\{0\}\cup\mT_1}x_s\alpha \geq \sum_{s\in\{0\}\cup\mT_1}x_s'\alpha$, $$\sum_{s\in\mT_2 \atop s\leq l}(d-\sum_{i=0}^{s-1}x_i)\beta=\sum_{s\in\mT_2 \atop s\leq l}(d-\sum_{i=0}^{s-1}x_i')\beta,$$
$$
\sum_{s\in\mT_2 \atop l<s\leq j}(d-\sum_{i=0}^{s-1}x_i)\beta\geq\sum_{s\in\mT_2 \atop l<s\leq j}(d-\sum_{i=0}^{s-1}x_i')\beta,
$$
and
$$
\sum_{s\in\mT_2 \atop j<s}(d-\sum_{i=0}^{s-1}x_i)\beta=\sum_{s\in\mT_2 \atop j<s}(d-\sum_{i=0}^{s-1}x_i^{\prime})\beta.
$$
Therefore,
\begin{equation*}
c_{\mT_1}(\bx) \geq c_{\mT_1}(\bx^{\prime}).
\end{equation*}
Hence, in these cases, if there are two positive components whose values are strictly less than $r$,
the value of $c_{\mT_1}(\bx)$ can be reduced by increasing the first component and decreasing the second by the same amount.
This procedure can be repeated until the first component rises to $r$ or the second one drops to 0. By repeating the argument, we can obtain an optimal solution that has no more than two components that are positive and strictly less than $r$.

In the remaining case where $l\in\{0\}\cup \mT_1,  j \in \mT_2$, we have
$$
c_{\mT_1}(\bx) - c_{\mT_1}(\bx^{\prime})=(x_l-x_l')\alpha-\sum_{s\in\mT_2 \atop l<s\leq j}(x_l-x_l')\beta.
$$
If $\alpha\geq \sum_{s\in\mT_2 \atop l<s\leq j}\beta$, we can set $x_l'=x_l-\delta$ and $x_j'=x_j+\delta$; otherwise, we can set $x_l'=x_l+\delta$ and $x_j'=x_j-\delta$. Thus we obtain $c_{\mT_1}(\bx) \geq c_{\mT_1}(\bx^{\prime}).$
By repeating the same argument as above, we can also conclude that, there is an optimal solution that has no more than two components that are positive and strictly less than $r$.
\end{proof}

\end{document}